\documentclass[11pt]{article}
\usepackage[top=1in,bottom=1in,right=1in,left=1in]{geometry}
\usepackage{amssymb,amsmath,amsthm,xcolor}
\definecolor{forestgreen}{rgb}{0.13, 0.55, 0.13}
\usepackage[colorlinks, linkcolor=red,citecolor=forestgreen,urlcolor = black, bookmarks=true]{hyperref}
 
\usepackage[nameinlink, noabbrev, capitalize]{cleveref}
\usepackage{tikz}

\usepackage{arydshln}
\usepackage{verbatim}
\usepackage{titlesec}
\usepackage{subfigure}
\usepackage[linesnumbered,ruled]{algorithm2e}
\usepackage{bm}
\usepackage{commath}
\allowdisplaybreaks

\titleclass{\subsubsubsection}{straight}[\subsection]

\newcounter{subsubsubsection}[subsubsection]
\renewcommand\thesubsubsubsection{\thesubsubsection.\arabic{subsubsubsection}}

\titleformat{\subsubsubsection}
  {\normalfont\normalsize\bfseries}{\thesubsubsubsection}{1em}{}
  \titlespacing*{\subsubsubsection}
  {0pt}{3.25ex plus 1ex minus .2ex}{1.5ex plus .2ex}

\makeatletter  
\def\toclevel@subsubsubsection{4}
\def\l@subsubsubsection{\@dottedtocline{4}{7em}{4em}}

\makeatother
\crefname{equation}{}{}
\crefrangeformat{equation}{(#3#1#4) --~(#5#2#6)}
\crefmultiformat{equation}{(#2#1#3)}{, (#2#1#3)}{, (#2#1#3)}{, (#2#1#3)}
\crefname{lemma}{Lemma}{Lemmas}
\crefname{section}{Section}{Sections}
\crefname{subsubsubsection}{Section}{Sections}
\crefname{remark}{Remark}{Remarks}
\crefname{figure}{Figure}{Figures}
\crefname{table}{Table}{Tables}
\Crefname{lemma}{Lemma}{Lemmas}
\crefname{theorem}{Theorem}{Theorems}
\Crefname{theorem}{Theorem}{Theorems}

\newcommand{\eps}{\varepsilon}

\newtheorem{theorem}{Theorem}[section]
\newtheorem{conjecture}[theorem]{Conjecture}

\newtheorem{remark}[theorem]{Remark}

\newtheorem{fact}[theorem]{Fact}
\newtheorem{lemma}[theorem]{Lemma}

\newtheorem{proposition}[theorem]{Proposition}
\newtheorem{corollary}[theorem]{Corollary}

\newtheorem{THM}[theorem]{Theorem}
\crefname{THM}{Theorem}{Theorems}

\theoremstyle{definition}

\theoremstyle{definition}
\newtheorem{definition}{Definition}[section]

\usepackage[suppress]{color-edits}
\addauthor{jg}{magenta}
\addauthor{as}{cyan}
\addauthor{ec}{forestgreen}
\addauthor{sl}{red}
\addauthor{chl}{blue}

\title{Fractional Pseudorandom Generators from Any Fourier Level}

\author{Eshan Chattopadhyay\thanks{Supported by NSF grant CCF-1849899.} \\
Cornell University\\
\texttt{eshanc@cornell.edu}
\and
Jason Gaitonde\thanks{Supported by NSF grant CCF-1408673 and AFOSR grant FA9550-19-1-0183.} \\
Cornell University\\
\texttt{jsg355@cornell.edu}
 \and
Chin Ho Lee\thanks{Supported by the Croucher Foundation and the Simons Collaboration on Algorithms and Geometry.} \\
Columbia University\\
\texttt{c.h.lee@columbia.edu}
\and
Shachar Lovett\thanks{Supprted by NSF grants CCF-2006443 and DMS-1953928.} \\
University of California, San Diego\\
\texttt{slovett@cs.ucsd.edu}
\and
Abhishek Shetty\thanks{Supported by a Cornell University Fellowship and a JP Morgan Chase Faculty Fellowship.} \\
Cornell University\\
\texttt{shetty@cs.cornell.edu}
}

\begin{document}
 \pagenumbering{gobble}

\maketitle
  \begin{abstract}
  We prove new results on the polarizing random walk framework introduced in recent works of Chattopadhyay {et al.} \cite{CHHL,CHLT} that exploit $L_1$ Fourier tail bounds for classes of Boolean functions to construct pseudorandom generators (PRGs).  We show that given a bound on the $k$-th level of the Fourier spectrum, one can construct a PRG with a seed length whose quality scales with $k$.  This interpolates previous works, which either require Fourier bounds on all levels \cite{CHHL}, or have polynomial dependence on the error parameter in the seed length \cite{CHLT}, and thus answers an open question in \cite{CHLT}.  As an example, we show that for polynomial error, Fourier bounds on the first $O(\log n)$ levels is sufficient to recover the seed length in \cite{CHHL}, which requires bounds on the entire tail.

We obtain our results by an alternate analysis of fractional PRGs using Taylor's theorem and bounding the degree-$k$ Lagrange remainder term using multilinearity and random restrictions.  Interestingly, our analysis relies only on the \emph{level-k unsigned Fourier sum}, which is  potentially a much smaller quantity than the $L_1$ notion in previous works.  By generalizing a connection established in \cite{xorlemma}, we give a new reduction from constructing PRGs to proving correlation bounds. Finally, using these improvements we show how to obtain a PRG for $\mathbb{F}_2$ polynomials with seed length close to the state-of-the-art construction due to Viola \cite{viola2009sum}, which was not known to be possible using this framework.
\end{abstract}

\newpage
\pagenumbering{arabic}
\section{Introduction}
A central pursuit in complexity theory is to understand the need of randomness in efficient computation. Indeed there are important conjectures (such as $\mathbf{P}=\mathbf{BPP}$) in complexity theory which state that one can completely remove the use of randomness without losing much in efficiency. While we are quite far from proving such results, a rich line of work has focused on \emph{derandomizing} simpler models of computation (see \cite{vadhan2012pseudorandomness} for a survey of prior work on derandomization). A key tool for proving such derandomization results is through the notion of a \emph{pseudorandom generator} defined as follows.
\begin{definition}\label{def:prg}
Let $\mathcal{F}$ be a class of $n$-variate Boolean functions. A \emph{pseudorandom generator} (PRG) for $\mathcal{F}$ with error $\eps>0$ is a random variable $\mathbf{X}\in \{-1,1\}^n$ such that for all $f\in \mathcal{F}$,
\begin{equation*}
    \bigl| \mathbb{E}_{\mathbf{X}}[f(\mathbf{X})]-\mathbb{E}_{\mathbf{U}_n}[f(\mathbf{U}_n)] \bigr| \leq \eps,
\end{equation*}
where $\mathbf{U}_n$ is the uniform distribution on $\{-1,1\}^n$.  We also say that $\mathbf{X}$ \emph{fools} $\mathcal{F}$ with error $\eps$.  If $\mathbf{X}=G(\mathbf{U}_s)$ for some explicit function $G:\{-1,1\}^s\to \{-1,1\}^n$, then $\mathbf{X}$ has \emph{seed length} $s$.
\end{definition} 
There is a long line of research on explicit constructions of PRGs for various classes of Boolean functions in the literature and it is well beyond our scope to survey prior work here. We focus on a recent line of works initiated by Chattopadhyay {et al.} \cite{CHHL, CHLT} that provide a framework for constructing pseudorandom generators for any Boolean function classes that exhibit \emph{Fourier tail bounds} (we will define and discuss this in more details in the next subsection; see \Cref{subsec:fourier} for a brief introduction to Fourier analysis of Boolean functions). This provides a unified PRG for several well-studied function classes such as small-depth circuits, low-sensitivity functions, and read-once branching programs that exhibit such Fourier tails.

We now briefly discuss this new framework, and then in \Cref{intro:results} we present our new results,  which significantly generalize this approach.

\subsection{The Polarizing Random Walk Framework} \label{intro:polarize}
The \emph{polarizing random walk} framework was introduced by Chattopadhyay, Hatami, Hosseini, and Lovett \cite{CHHL}. 
The authors showed that for any class of $n$-variate Boolean functions that are closed under restrictions, one can flexibly construct pseudorandom generators via the following local-to-global principle: it suffices to construct \emph{fractional pseudorandom generators (fractional PRGs)}, a notion that generalizes PRGs to allow the random variable $\mathbf{X}$ (in \Cref{def:prg}) to be supported on the solid cube $[-1,1]^n$ instead of $\{-1,1\}^n$, while still requiring that $\mathbf{X}$ fools (the multilinear extension) of each Boolean function in the class.
Ideally, the variance of each coordinate of $\mathbf{X}$ should be as large as possible.  
Towards this,  we define a fractional PRG $\mathbf{X}$ to be $p$-noticeable if the variance in each of its coordinates is least $p$ (See \Cref{def:fprg} for a formal definition of a fractional PRG).

To obtain a genuine pseudorandom generator from a fractional PRG,
the authors give a random walk gadget that composes together independent copies of the fractional PRG in a random walk that polarizes  $\mathbf{X}$ quickly to take values from the Boolean hypercube $\{-1,1\}^n$.
The analysis for how the error accumulates in this process relies on interpreting the intermediate points of $\mathbf{X}$ in this 
random walk as an average of \emph{random restrictions} of the original Boolean function. As the fractional PRG locally fools the class by definition, this analysis shows that the random walk does not incur much additional error at each intermediate step and the rapid polarization shows that it does not take too many steps.  Taken together, these two facts imply that the final random variable (supported on  $\{-1,1\}^n$) successfully fools the class. 

\jgdelete{The above framework shows that if one can construct  non-Boolean random variables with sufficiently large variance in each coordinate, then one immediately obtains a pseudorandom generator using their random walk gadget.} 
Through this construction, the design of pseudorandom generators reduces to the easier task of designing fractional pseudorandom generators.
It is easier as such random variables need not be Boolean-valued.
 The authors further construct such fractional pseudorandom generators for any class of functions satisfying \emph{Fourier tail bounds}, that is, every function in the class is such that the $L_1$ Fourier mass at each level $1\leq k\leq n$ is at most $b^k$ for some fixed $b\geq 1$.
 For error $\eps$, their fractional pseudorandom generators have seed length $O(\log\log n+ \log(1/\eps))$ and variance $\Theta(b^{-2})$ in each coordinate.  Combining this fractional pseudorandom generator with their random walk gadget yields a pseudorandom generator with seed length $b^2\cdot \mathrm{polylog}(n/\eps)$ for \emph{any} class with such Fourier tail bounds. 

As a result, if one can show that a function class admits nontrivial Fourier tail bounds (and is closed under restriction), then the \cite{CHHL} construction immediately implies a pseudorandom generator for this class. Some examples of Boolean functions that exhibit such tail bounds include $\mathbf{AC}^0$ circuits with the parameter $b =\mathrm{poly}(\log n)$ \cite{LMN89,tal2017tight}, constant width read-once branching programs with $b =\mathrm{poly}(\log n)$ \cite{chattopadhyay2018improved}, $s$-sensitive functions with $b=O(s)$ \cite{GSW16,GSTW}, and product tests~\cite{DBLP:conf/coco/Lee19}. Using these tail bounds, \cite{CHHL} immediately gave  
PRGs for these function classes. It was also conjectured in \cite{CHHL} that the class of $n$-variate degree-$d$ polynomials over $\mathbb{F}_2$ satisfy such tail bounds.  We discuss this in more details in \Cref{intro:results}.

A natural question is whether the complete control on the entire Fourier tail of a class is necessary to obtain a PRG in this framework. In the subsequent work by Chattopadhyay, Hatami, Lovett, and Tal \cite{CHLT}, the authors show how to construct fractional pseudorandom generators using different pseudorandom primitives whose seed length depends on just the \emph{second Fourier level} of the class. They construct their fractional PRGs by derandomizing the celebrated work of Raz and Tal \cite{RazTal}, which establishes an oracle separation of $\mathbf{BQP}$ and $\mathbf{PH}$.  Raz and Tal show that classes of multilinear functions with small level-two Fourier mass cannot significantly distinguish between a suitable variant of the Forrelation distribution and the uniform distribution.\footnote{It turns out that this fact can be interpreted via It\^o's Lemma, which shows that the local behavior of a smooth function of Brownian motion is essentially determined by the first two derivatives~\cite{wu2020stochastic}.}
\jgdelete{\cite{CHLT} show that one can derandomize this analysis by efficiently constructing fractional PRGs that simulate Gaussian random variables with small covariance using the best-known constructions of error-correcting codes.}
However, this construction incurs exponentially worse dependence on the error parameter in each fractional step to sample sufficiently good approximations to Gaussian random variables. The final seed length given by this construction has the form $O((b^2/\eps)^{2+o(1)} \mathrm{polylog}(n))$, where $b^2$ is the level-two Fourier mass of the class. This yields exponentially worse dependence on the error compared to the generator of \cite{CHHL}, as well as quadratically worse dependence on the level-two mass (though {without assumptions on the} rest of the Fourier levels).

\subsection{Our Contribution}\label{intro:results}
{In this paper, we address several open questions in this framework by leveraging a novel connection between polarizing random walk and the classical theory of polynomial approximation.} Given these {prior} works, a very natural question (also explicitly asked in \cite{CHLT}) is whether it is possible to interpolate between these {previous} constructions by assuming Fourier bounds on an intermediate level. Concretely, can this framework still succeed if one has Fourier control at just level $k$? If the class further has such Fourier bounds up to and including level $k$, can one interpolate between the seed lengths of \cite{CHHL} and \cite{CHLT}? Given Fourier bounds from level $1$ up to level $k$, what range of error $\eps>0$ can the resulting PRG tolerate while maintaining polylogarithmic dependence on $1/\eps$ in the seed length (or equivalently, given a desired error $\eps>0$, how many levels of Fourier bounds are sufficient to ensure that the seed length remains polylogarithmic in $1/\eps$)?

Moreover, {it was previously not known whether $L_1$ control of Fourier tails is really necessary for this framework to yield effective PRGs, or whether weaker Fourier quantities would suffice. To this end, define
\begin{equation*}
    L_{1,k}(f)\triangleq \sum_{S\subseteq [n]:\vert S\vert=k} \vert \hat{f}(S)\vert
\end{equation*}
to be the \emph{level-$k$ $L_1$ Fourier mass} of $f$, and 
\begin{equation*}
    M_k(f)\triangleq \max_{\mathbf{x}\in [-1,1]^n}\bigg\vert \sum_{S:\vert S\vert=k} \hat{f}(S)\mathbf{x}^S\bigg\vert =  \max_{\mathbf{x}\in \{-1,1\}^n}\bigg\vert \sum_{S:\vert S\vert=k} \hat{f}(S)\mathbf{x}^S\bigg\vert .
\end{equation*}
to be the \emph{level-$k$ absolute Fourier sum} of $f$.
For a function class $\mathcal{F}$, we define $L_{1,k}(\mathcal{F})$ and $M_k( \mathcal{F} )$ as the maximum of $L_{1,k}(f)$ and $M_k(f)$ taken over $f\in \mathcal{F}$.}  The recent work by Chattopadhyay, Hatami, Hosseini, Lovett, and Zuckerman \cite{xorlemma} considers the weaker quantity of the level-two {\em unsigned Fourier sum}, defined as the absolute value of the sum of the Fourier coefficients rather than the sum of their absolute values that is considered in \cite{CHHL,CHLT}.  The authors show that the problem of bounding the level-two unsigned Fourier sum corresponds to the problem of bounding the covariance of the function class and the $\mathsf{XOR}$ of shifted majority functions. {For a class that is closed under negations of the variables, the level-two unsigned Fourier sum is precisely the quantity $M_2(\mathcal{F})$.} In particular, using this connection to this weaker object, the authors explicitly ask whether bounding the weaker Fourier quantity $M_2(\mathcal{F})$ (or more generally, $M_k(\mathcal{F})$) yield pseudorandom generators.

{In this work, we positively resolve all of these questions. To do so, we establish novel connections between the polarizing random walk framework and the classical theory of polynomial approximations of Boolean functions. We show that the seed length of a fractional PRG for a given class of functions $\mathcal{F}$ is intimately connected to the uniform error of low-degree approximations of functions on \emph{subcubes} of the form $[-c,c]^n$ for some $c<1$.} 

{Our main technical result provides an upper bound on this quantity in terms of $M_k(\mathcal{F})$ for every function $f$ in a class $\mathcal{F}$ that is closed under restrictions. For any multilinear polynomial $f:\{-1,1\}^n\to \mathbb{R}$, define $f_{\geq k}$ to be component of $f$ with monomials of degree at least $k$. Then our main result asserts the following bound:}

\begin{THM}
\label{thm:upperbound}
        Let $f\in \mathcal{F}$ with $\mathcal{F}$ closed under restrictions.  Then for all $c \in (0,1)$, we have
    \begin{equation*}
        \max_{\mathbf{x} \in [-c , c]^n } \vert f_{\geq k} (\mathbf{x})\vert \leq \left( \frac{c}{1 - c} \right)^k M_k(\mathcal{F}).
    \end{equation*}
\end{THM}

{For intuition, recall that by Parseval's identity in Fourier analysis the low-degree Fourier expansion of any Boolean function $f$ is provably the best $\ell_2$-approximator on $\{-1,1\}^n$. Conversely, from elementary analysis, one can show that the best uniform (i.e. $\ell_{\infty}$) low-degree approximators of $f$ converge, coefficient-by-coefficient, to the low-degree expansion of $f$ as the domain converges to $\mathbf{0}$. Our main result shows that one can strongly quantify the $\ell_{\infty}$ error of the low-degree approximator of Boolean functions on subcubes so long as $c$ is not too close to $1$ (compare this bound to when $f$ has degree exactly $k$).}

{We complement this result with a corresponding lower bound on the best attainable uniform error for \emph{any} low-degree approximation on these subcubes that will be comparable for sufficiently small values of $c$ (see \Cref{thm:cheby}).  These results combined together imply that the low-order expansion of a Boolean function is a reasonable uniform approximation for small domains.  Note that the properties of low-degree approximations on subcubes with $c\ll 1$ can be quite different than for $c=1$; for instance the $\mathsf{PARITY}$ function on $n$ bits is well-known to be inapproximable on $\{-1,1\}^n$ to constant error unless the approximating polynomial has degree $\Omega(n)$, but is trivially approximable for any $c$ bounded away from $1$.}

{From this main result, we can positively resolve the above open questions in the polarizing random walk framework as a nearly immediate corollary. To do so, we provide a new analysis of the fractional pseudorandom generator of \cite{CHHL} that views fractional pseudorandom generators as fooling a low-degree part of a function on $[-c,c]^n$ for some $c<1$, where the high-degree part has small $\ell_{\infty}$ norm on $[-c,c]^n$. Recall that the seed length of the final generator depends on the variance of the constituent fractional generator; the connection to the above result is that for a given error $\eps$, the largest subcube on which the above approximation holds can be lower-bounded using just the weaker $M_k(\mathcal{F})$ quantity.}
Leveraging this insight, our main result in the polarizing random walk framework is the following analysis of a fractional pseudorandom generator:
\begin{THM}
\label{thm:fracprg1}
Let $\mathcal{F}$ be any class of $n$-variate Boolean functions that is closed under restrictions.  Suppose $M_k(\mathcal{F})\leq b^{k}$ for some $b\geq 1$ and  $k \geq 1$. Then for any $\eps>0$, there exists an explicit $\Omega(\eps^{2/k}/b^2)$-noticeable fractional PRG for $\mathcal{F}$ with error $\eps$ and seed length $O(k\cdot \log n)$.\footnote{We remark that at this level of generality, this linear dependence on $k$ is essentially necessary. Indeed, any Boolean function on $n$-variables has $L_1$ level-$n$ mass at most $1$, but one cannot hope to generically fool all Boolean functions simultaneously without using $n$ bits.}

Further, if it holds that $L_{1,i}(\mathcal{F})\leq b^{i}$ for all $1\leq i<k$, then the seed length can be improved to $O(\log\log n + \log k+\log(1/\eps))$.
\end{THM}

Using the fractional pseudorandom generator from \Cref{thm:fracprg1}, we obtain the following consequences almost immediately from the random walk gadget of \cite{CHHL} (see \cref{thm:amplification}):

\begin{enumerate}
\item \textbf{Pseudorandom Generators from Fourier Bounds at Level $k$}: From our fractional pseudorandom generator, we show that the random walk framework yields nontrivial pseudorandom generators assuming Fourier bounds \emph{just at} level $k$ of the associated class, with improvements if we assume bounds from level $1$ \emph{up to} level $k$. The informal statement is the following:

\begin{THM}
\label{thm:levelkprg}
Let $\mathcal{F}$ be any class of $n$-variate Boolean functions that is closed under restrictions. Suppose that $\mathcal{F}$ satisfies $M_{k}(\mathcal{F})\leq b^{k}$ for some  $b\geq 1$ and  $k>2$. Then there exists an explicit pseudorandom generator for $\mathcal{F}$  for error $\eps$ with seed length $k\cdot b^{2+4/(k-2)}\mathrm{polylog}(n/\eps)/\eps^{2/(k-2)}$. The seed length can be improved if $L_{1,i}(\mathcal{F}) \le b^i$ for all levels $i \le k$.
\end{THM}

 See \cref{thm:combined} for the precise statement. One immediate consequence is that if one has a non-trivial bound on $M_3(\mathcal{F})$, then the seed length of our PRG has the same dependence on the error $\eps$ as the one in \cite{CHLT}.  Further, given $M_4(\mathcal{F}) \le b^4$, one obtains better seed length than \cite{CHLT}; in particular it has quadratically better dependence on $1/\eps$ in the seed length (as well as polylogarithmic factors in $n/\eps$).  More generally, given an appropriate Fourier bound of $b^k$ on just some level $k\leq \mathrm{polylog}(n)$, one obtains a pseudorandom generator with error $\eps$ with seed length $O(b^{2+4/(k-2)}\mathrm{polylog}(n/\eps)/\eps^{2/(k-2)})$. 
 
We note that the fractional PRG from \Cref{thm:fracprg1} cannot be converted into a PRG for $k=1,2$. Informally, this is because of the following reason: the number of steps one needs to take in the random walk gadget of \cite{CHHL} (with each step using an independent copy of the fractional PRG) scales roughly with the variance of the fractional PRG, and the error adds up in each step.  As is clear from \Cref{thm:fracprg1}, for the variance of the fractional PRG to scale sublinearly with the error, one requires $k>2$. See \cref{rem:lev-2} for more discussion. \jgcomment{What do people think about deleting this? While accurate, it seems to put a slight damper on the results and is addressed later on.}

\item \textbf{Pseudorandom Generators with Polylogarithmic Error Dependence from Up-to-level-$k$ Bounds}: A simple corollary of our fractional pseudorandom generator is that one can recover the polylogarithmic dependence on $1/\eps$ from \cite{CHHL} if $\eps \geq b\cdot \log n \cdot 2^{-O(k)}$ and we have Fourier bounds \emph{up to} level $k$. 
\begin{corollary}
\label{cor:polylogerr}
Let $\mathcal{F}$ be any class of $n$-variate Boolean functions that is closed under restrictions.  Suppose that for some level $k>2$ and $b\geq 1$, we have $M_k(\mathcal{F})\leq b^k$ and  $L_{1,i}(\mathcal{F})\leq b^{i}$ for $i<k$. Then, for any $\eps\geq b\cdot\log n \cdot 2^{-O(k)}$, there exists an explicit pseudorandom generator for $\mathcal{F}$ with error $\eps$ and seed length $O( b^{2} \mathrm{polylog}(n/\eps))$.
\end{corollary}
 
  This actually subsumes the analysis of \cite{CHHL} without requiring anything on the full Fourier tail, and addresses an open question of \cite{CHLT} asking how many levels of Fourier bounds one needs control of to regain polylogarithmic dependence on $\eps$. In particular, if one requires error $\eps = 1/\mathrm{poly}(n)$, then it suffices to have Fourier bounds up to level $\Theta(\log n)$ to get the same dependence.
\end{enumerate}
We view this work as a proof of concept that it is indeed possible to interpolate between the two extremes of \cite{CHHL,CHLT} in the polarizing random walk framework and obtain better results using weakened Fourier assumptions.  We prove \Cref{thm:fracprg1} in \Cref{sec:patoprg}, from which \Cref{thm:levelkprg} and \Cref{cor:polylogerr} follow without much difficulty using the existing random walk gadget of \cite{CHHL}. 

Note that for some Boolean classes of great interest such as the class of low-degree $\mathbb{F}_2$-polynomials, Fourier tail bounds as required by \cite{CHHL} are not yet known and thus \Cref{thm:fracprg1} allows us to leverage potentially much weaker bounds proved in \cite{CHHL} to construct a PRG with polylogarithmic dependence on $n/\eps$ in the seed length (see \Cref{thm:main_poly}).  This almost matches the best known PRG due to Viola \cite{viola2009sum}. In particular, we show the following:

\begin{THM}\label{thm:main_poly}
Let $\mathcal{F}$ be the class of degree-$d$ polynomials over $\mathbb{F}_2$ on $n$ variables. Then there exists an explicit pseudorandom generator for $\mathcal{F}$ with error $\eps$ and seed length $2^{O(d)}\mathrm{polylog}(n/\eps)$.
\end{THM}

We present the proof of \Cref{thm:main_poly} in \Cref{sec:apps}. While this result does not quite match the current state-of-the-art PRG for this class due to Viola \cite{viola2009sum}  (and therefore fails to give anything nontrivial for $d=\Omega(\log n)$),  we view this as a conceptual contribution  that the random walk framework can yield an explicit pseudorandom generator with seed length that is polylogarithmic in $n/\eps$, which was not known from previous works~\cite{CHHL,CHLT}.  As we discuss below, the results in \cite{CHHL,CHLT} do not give a PRG for the class of $\mathbb{F}_2$-polynomials using known Fourier tail bounds. 

As a concrete application of this approach which would dramatically improve the state-of-the-art PRGs for $\mathbb{F}_2$-polynomials, both \cite{CHHL} and \cite{CHLT} conjecture Fourier bounds on the $L_1$ mass of the class of degree-$d$ $\mathbb{F}_2$ polynomials. The former conjectures that this class satisfies a tail bound of the form $c_d^k$ for some constant $c_d$ at all levels $1\leq k\leq n$ (so as to apply their approach), while the latter conjectures just that the level-two $L_1$ mass is $O(d^2)$.  While neither conjecture seems close to being resolved, {our work shows that one can instead prove bounds for the smaller quantities $M_k(\mathcal{F})$ for any $k\geq 3$.} If one could prove such bounds of the form $(\text{poly}(d,\log n))^k$ for some level $k=\Omega(1)$, or even more optimistically, for some $k=\Omega(\log n)$, this would immediately imply a breakthrough pseudorandom generator for $\mathbf{AC^0[\oplus]}$ using the results of Razborov \cite{Razborov1987} and Smolensky \cite{Smolensky1,Smolensky2} (see the discussion in \cite{CHLT}).

\jgdelete{Moreover, as stated before, recent work \cite{xorlemma} has shown that it is possible to deduce bounds on $M_2(\mathcal{F})$ using covariance bounds with the $\mathsf{XOR}$ of certain resilient functions. As we are able to show that bounds on such quantities imply pseudorandom generators, we give an analogous argument for an appropriate generalization of this result to $M_k(\mathcal{F})$ in \cref{sec:corrsec}, thus reducing the problem of constructing PRGs in this framework to proving correlation bounds.}
{To our knowledge, our application of $M_k(\mathcal{F})$ bounds is new to the pseudorandomness literature.  There are several advantages to proving $M_k(\mathcal{F})$ bounds over $L_{1,k}(\mathcal{F})$ bounds. For one, from the definition we clearly have $M_{k}(\mathcal{F})\leq L_{1,k}(\mathcal{F})$ for any class $\mathcal{F}$.  This improvement alone potentially gives smaller seed length for any class. From an analytical perspective, we believe that the quantity $M_k(\mathcal{F})$ is easier to estimate.  Specifically, for a class $\mathcal{F}$ that is closed under negation of input variables, $M_k(\mathcal{F})$ is precisely an \emph{unsigned
Fourier sum} and can be bounded via the recent connections established by Chattopadhyay {et al.}~\cite{xorlemma}, which reduces $M_2(\mathcal{F})$ bounds to proving correlation bounds against certain resilient functions.  We straightforwardly generalize their reduction to $M_k(\mathcal{F})$ bounds in \Cref{sec:corrsec}.}

\subsection{Overview of Our Approach}
{To prove \Cref{thm:upperbound}, we rely on Taylor's Theorem, as well as multilinearity and the random restriction trick of \cite{CHHL}. Recall that Taylor's Theorem, when applied to a sufficiently smooth function $h\colon[-1,1]\to \mathbb{R}$, asserts that the Taylor expansion at $0$ can be expressed in terms of its first $(k-1)$-th order derivatives at $0$ along with a Lagrange error term that depends on its $k$-th order derivatives at some intermediate point in our domain.  In doing so, the higher-order components of the function ``collapse'' down to the $k$-th order term.  While Taylor's Theorem has been extensively applied in the construction of pseudorandom generators, often in tandem with \emph{invariance principles}, we somewhat counterintuitively apply it to the \emph{multilinear expansion of the Boolean functions} themselves.} 

{To apply Taylor's theorem here, we consider one-dimensional restrictions of (the multilinear extension) of a Boolean function $f\colon\{-1,1\}^n \to \{-1,1\}$. While the full Taylor expansion of a polynomial is trivially the same polynomial, the Lagrange error term  eliminates the dependence on the high order Fourier coefficients (corresponding to the terms of degree $>k$). Moreover, the low-order terms of the  Taylor expansion of $f$ at $0$
are precisely the original low-degree part of its Fourier expansion.  However, the Lagrange error term requires the derivatives to be evaluated at a point away from $0$.
While the derivatives of $f$ at a nonzero point are related to the \emph{biased} Fourier coefficients of $f$, it is not clear how to estimate these quantities.
To overcome this difficulty, recall that we are interested in bounds on $\vert f_{\geq k}(\mathbf{x})\vert$ for $\mathbf{x}\in \{-c,c\}^n$ where $c<1$.  In \Cref{lemma:finallem}, we show that by ``recentering'' $\mathbf{x}$ using the random restriction technique of \cite{CHHL}, we can write the error term as an average of the $k$-th order derivatives \emph{at 0} of some random restrictions of our original function $f$, up to a multiplicative factor depending on $c$.  We can then apply multilinearity to bound these error terms using $M_k(\mathcal{F})$ to obtain \Cref{thm:main}.}

{While \Cref{thm:main} shows that the low-order Taylor expansion of a Boolean function is a decent \emph{uniform} approximator on subcubes $[-c,c]^n$ for some sufficiently small $c$ that depends on the class $\mathcal{F}$, it is natural to wonder if one can obtain a better 
low-order approximation.  Using our upper bound along with Chebyshev polynomials on the univariate restrictions, we give a lower bound
showing that no low-order approximator can give significantly smaller error over $[-c,c]^n$ for any $c$ less than some quantity depending on the ratio $M_k(\mathcal{F})/M_{k+1}(\mathcal{F})$ for some $k$. This quantifies the intuition that the low-degree Fourier expansion is a near optimal uniform approximator of $f$ over small enough neighborhoods of $\mathbf{0}$. These arguments are formally carried out in \Cref{sec:lowdeg}.}

To prove our results {in the polarizing random walk framework}, we rely on an alternate, simple analysis of fractional pseudorandom generators.
{The original analysis in \cite{CHHL} assumes control of $L_{1,k}(\mathcal{F})$ at all levels of the Fourier spectrum.
We now explain how these assumptions can be weakened using \Cref{thm:main}.
Consider a candidate fractional PRG $\mathbf{X}\in [-1,1]^n$.} 
We first decompose the multilinear (Fourier) expansion of $f\in \mathcal{F}$
in the same manner as \cite{CHHL}:
\begin{equation}
   \bigl| \mathbb{E}_{\mathbf{X}}[f(\mathbf{X})]-\mathbb{E}_{\mathbf{U}}[f(\mathbf{U})] \bigr|
   \leq  \underbrace{\sum_{i=1}^{k-1}\sum_{S\subseteq [n]:\vert S\vert = i} \bigl| \hat{f}(S) \bigr|
   \bigl| \mathbb{E}_{\mathbf{X}}[\mathbf{X}^S] \bigr|}_{\text{low-order terms}} + 
   \underbrace{\vphantom{\sum_{S\subseteq[n}} \bigl| \mathbb{E}_{\mathbf{X}}[f_{\geq k}(\mathbf{X})] \bigr| }_{\text{high-order term}}.
\end{equation}
{\cite{CHHL} requires bounding $L_{1,\ell}(\mathcal{F})$ for all $\ell \ge k$ to give a uniform bound on the high-order term.  
Using \Cref{thm:main}, we can obtain small error in the high-order term so long as we choose $\mathbf{X}\in [-c,c]^n$ for sufficiently small $c$ depending on $\eps$ and $M_k(\mathcal{F})$. 
To handle the low-order terms, we consider two cases: if we further have $L_{1,\ell}(\mathcal{F})$ bounds for $\ell<k$, then we may choose $\mathbf{X}$ to be a scaled $\delta$-almost $(k-1)$-wise distribution to nearly fool each of the low-order terms as in \cite{CHHL}. Otherwise, we may choose $\mathbf{X}$ to be a scaled $(k-1)$-wise independent distribution to incur zero error from the low-order terms. Note that the latter pseudorandom primitives are more expensive in terms of seed length.} Finally, to obtain pseudorandom generators, we then simply apply the random walk gadget of \cite{CHHL} to our fractional PRGs as a blackbox.  We refer the reader to \Cref{sec:patoprg} for formal proofs of the ideas in this section.

{We immediately leverage this newfound flexibility to construct new pseudorandom generators for $\mathbb{F}_2$-polynomials of degree $d = O(\log n)$.
We do this using known $L_{1,k}(\mathcal{F})$ bounds derived in \cite{CHHL}.
Previously these bounds were not sufficient to give PRGs as their analysis of fractional PRGs requires control of the entire Fourier tail,
but they can be employed here as our analysis no longer requires so. This result is given in \Cref{sec:apps}. Finally, we show how $M_k(\mathcal{F})$ bounds can be obtained using correlation bounds with shifted majority functions in \Cref{sec:corrsec}. This is done by straightforwardly generalizing the analysis of \cite{xorlemma}, which shows how such correlation bounds can be used to bound the bulk of the terms in the definition of $M_k(\mathcal{F})$.}

\subsection{Other Related Work}
{To our knowledge, our use of $M_k(\mathcal{F})$ bounds is new to the derandomization literature.  As mentioned earlier, the stronger and better-known $L_{1,k}(\mathcal{F})$ notion has been extensively studied in recent years.  In addition to derandomization, a recent line of work~\cite{DBLP:journals/eccc/Tal19,DBLP:journals/corr/abs-2008-07003,DBLP:journals/eccc/SherstovSW20} has used $L_{1,k}$ bounds for decision trees to
obtain an optimal separation of quantum and classical query complexity.  
Among these works, the work of Bansal and Sinha~\cite{DBLP:journals/corr/abs-2008-07003} generalizes the results of Raz and Tal \cite{RazTal} by considering a $k$-generalization of their Forrelation distribution and bounding the distinguishing advantage of any function with small $L_{1,\ell}$ bounds for $\ell=1,\ldots,k$. 
Much as how the results of Chattopadhyay {et al.} \cite{CHLT} derandomize the result of Raz and Tal, we believe that their construction can be derandomized for pseudorandomness purposes, but appears to give significantly worse seed length, nor obtains bounds in terms of $M_k(\mathcal{F})$. A related work by Girish, Raz, and Zhan \cite{RazXOR} establishes a similar result with a different generalization of the Forrelation distribution, but we do not know how to use their construction for pseudorandom generators.}


{The relationship between $M_k(\mathcal{F})$ and $L_{1,k}(\mathcal{F})$ has been of intense study in the mathematics literature due to renewed interest in \emph{Bohnenblust--Hille} inequalities (see, for instance, the breakthrough work of Defant, Frerick, Ortega-Cerd\`a, Ouna\"ies, and Seip \cite{defant2011bohnenblust}). The optimal constant $C_{n,k}$ satisfying $L_{1,k}(f)\leq C_{n,k}M_k(f)$ for any polynomial polynomial $f\colon\mathbb{C}^n\to \mathbb{C}$ is known as the \emph{Sidon constant}. It is known that $C_{n,k}$ is, up to small exponential factors in $k$, proportional to roughly $n^{\frac{k-1}{2}}$, and its tightness is witnessed by a random function with high probability.
The quantity $M_k(\mathcal{F})$ also has applications in other areas in theoretical computer science, such as quantum information theory (see for instance the survey of Montanaro \cite{montanaro2012some}) and Boolean function analysis \cite{DBLP:conf/stacs/Arunachalam0KSW20}.}

\section{Preliminaries}
\label{sec:preliminaries}
As in \cite{CHHL} and \cite{CHLT}, we study PRGs for classes $\mathcal{F}$ of $n$-variate Boolean functions that are closed under restriction (that is, fixing any subset of the input variables of a function in the class yields a function that remains in the class).

\subsection{Fourier Analysis} \label{subsec:fourier}
We briefly recall basic Fourier analysis: any Boolean function $f:\{-1,1\}^n \to \{-1,1\}$ admits a unique multilinear expansion, also known as the \emph{Fourier expansion}, given by 
\begin{equation}
\label{eq:fourierexpansion}
    f(\mathbf{x})=\sum_{S\subseteq [n]} \hat{f}(S)\mathbf{x}^S,
\end{equation}
where we write $\mathbf{x}^S\triangleq \prod_{i\in S} x_i$. The   Fourier coefficient $\hat{f}(S)$ is given by
\begin{equation*}
    \hat{f}(S)=\mathbb{E}_{\mathbf{X}\sim \{-1,1\}^n}[f(\mathbf{X})\mathbf{X}^S].
\end{equation*}
For more on Fourier analysis of Boolean functions, see the excellent book by O'Donnell \cite{ODonnelBook}. One may thus extend the domain of $f$ to $[-1,1]^n$, where $f(\mathbf{x})$ for arbitrary $\mathbf{x}$ is evaluated according to the expression in \Cref{eq:fourierexpansion}. Note that in this case, $f(\mathbf{0})=\hat{f}(\emptyset)=\mathbb{E}_{\mathbf{U}_n}[f(\mathbf{U}_n)]$. One of the main parameters of interest from the Fourier expansion for this framework is the following:

\begin{definition}
The \emph{level-$k$ mass of a Boolean function} $f$ is 
\begin{equation*}
    L_{1,k}(f)\triangleq \sum_{S\subseteq [n]:\vert S\vert=k} \vert \hat{f}(S)\vert,
\end{equation*}
and the \emph{level-$k$ mass of a class $\mathcal{F}$} is
$L_{1,k}(\mathcal{F})\triangleq \max_{f\in \mathcal{F}} L_{1,k}(f)$.
\end{definition}

In this work, we will show how to construct PRGs whose seed length depends on the following, smaller quantity:

\begin{definition}
\label{def:mk}
For any multilinear polynomial $f:\mathbb{R}^n\to \mathbb{R}$ given by $f(\mathbf{x})=\sum_{S\subseteq [n]} \hat{f}(S)\mathbf{x}^S$, define the level-$k$ part by
\begin{equation*}
    f_{k}(\mathbf{x})\triangleq\sum_{S\subseteq [n]:\vert S\vert=k} \hat{f}(S)\mathbf{x}^S,
\end{equation*}
and further define $f_{<k}(\mathbf{x})\triangleq \sum_{i=0}^{k-1} f_i(\mathbf{x})$ and $f_{\geq k}(\mathbf{x})\triangleq \sum_{i=k}^{n} f_i(\mathbf{x})$. 
Then we define the \emph{level-$k$ absolute Fourier sum of $f$} by
\begin{equation*}
    M_k(f)\triangleq \max_{\mathbf{x}\in [-1,1]^n}\bigg\vert \sum_{S:\vert S\vert=k} \hat{f}(S)\mathbf{x}^S\bigg\vert =  \max_{\mathbf{x}\in \{-1,1\}^n}\bigg\vert \sum_{S:\vert S\vert=k} \hat{f}(S)\mathbf{x}^S\bigg\vert
\end{equation*}
and analogously define $M_k(\mathcal{F})\triangleq \max_{f\in \mathcal{F}} M_k(f)$ for a class $\mathcal{F}$.
\end{definition}
Note that the equality arises by multilinearity, and clearly we have $M_k(f)\leq L_{1,k}(f)$ by the triangle inequality. Without loss of generality, we may further assume that our class is closed under flipping the image, i.e. we may suppose that $f\in \mathcal{F}$ if and only if $-f\in \mathcal{F}$; this transformation does not change either $L_{1,k}(f)$ or $M_k(f)$, and therefore the same bound on the class still holds when completing it to include all such functions. If this is the case, we get the more striking identity:

\begin{lemma}
\label{lem:unsigned}
Suppose that $\mathcal{F}$ is closed under negation of variables and that $f\in \mathcal{F}$ implies $-f\in \mathcal{F}$. Then
\begin{equation*}
    M_k(\mathcal{F})=\max_{f\in \mathcal{F}} \sum_{S:\vert S\vert = k} \hat{f}(S) = \max_{f \in  \mathcal{F} } f_k \left( \mathbf{1} \right) .
\end{equation*}
\end{lemma}

To see why this holds, simply note that if $(f,\mathbf{z})\in \mathcal{F}\times \{-1,1\}^n$ is a maximizer in the definition of $M_k(\mathcal{F})$ (where we may now assume that the sign is positive), then by replacing the function $f(\mathbf{x})$ with $g(\mathbf{x})=f(\mathbf{x}\circ \mathbf{z})$, where $\circ$ denotes componentwise multiplication, we have 
\begin{equation*}
    M_k(\mathcal{F})
    = \Biggl| \sum_{S:\vert S\vert=k} \hat{f}(S)\mathbf{z}^S \Biggr|
    = \sum_{S:\vert S\vert=k} \hat{g}(S)=\max_{h\in \mathcal{F}} \sum_{S:\vert S\vert=k} \hat{h}(S).
\end{equation*}
In particular, it suffices to bound the \emph{unsigned level-$k$ Fourier sum} of such a class.

Lastly, we require the following notion:

\begin{definition}
Let $\mathcal{F}$ be a class of $n$-variate multilinear polynomials that is closed under restrictions.  Define $\mathrm{conv}(\mathcal{F})$ as the convex closure of $\mathcal{F}$,
\begin{equation*}
    \mathrm{conv}(\mathcal{F})\triangleq \left\{\sum_{f\in \mathcal{F}} \lambda_f f\,\bigg\vert \sum_{f\in \mathcal{F}} \lambda_f=1, \lambda_f\geq 0 \,\,\forall f\in \mathcal{F}\right\}.
\end{equation*}
\end{definition}
We briefly note the following two elementary facts: first, by the assumption that $\mathcal{F}$ is closed under restrictions, the same is true of $\mathrm{conv}(\mathcal{F})$. The second is the following simple claim:

\begin{lemma}
For any class $\mathcal{F}$ of Boolean functions, $M_k(\mathcal{F})=M_k(\mathrm{conv}(\mathcal{F}))$.
\end{lemma}
\begin{proof}
One direction is obvious: as $\mathcal{F}\subseteq \mathrm{conv}\mathcal{F}$, clearly $M_k(\mathcal{F})\leq M_k(\mathrm{conv}(\mathcal{F}))$. In the other direction, let $g=\sum_{f\in \mathcal{F}} \lambda_f f$ be an arbitrary element of $\mathrm{conv}(\mathcal{F})$, where $\lambda_f\geq 0$ and $\sum_{f\in \mathcal{F}}\lambda_f=1$. Then
\begin{align*}
    M_k(g) &= \max_{\mathbf{x}\in \{-1,1\}^n} \Biggl| \sum_{S\subseteq [n]:\vert S\vert=k} \widehat{g}(S) \mathbf{x}^S \Biggr| \\
    &=\max_{\mathbf{x}\in \{-1,1\}^n} \Biggl| \sum_{S\subseteq [n]:\vert S\vert=k} \Bigl( \sum_{f\in \mathcal{F}}\lambda_f\widehat{f}(S) \Bigr) \mathbf{x}^S \Biggr| \\
    &\leq \sum_{f\in \mathcal{F}} \lambda_f \max_{\mathbf{x}\in \{-1,1\}^n} \Biggl| \sum_{S\subseteq [n]:\vert S\vert=k} \widehat{f}(S) \mathbf{x}^S \Biggr| \\
    &\leq \max_{f\in \mathcal{F}} M_k(f).
\end{align*}
The reverse inequality immediately follows.
\end{proof}

\subsection{(Fractional) Pseudorandom Generators}
We now recall the (well-known) definition of a pseudorandom generator, as well as the generalization of a fractional pseudorandom generator as introduced by \cite{CHHL}:

\begin{definition}
Let $\mathcal{F}$ be a class of $n$-variate Boolean functions. Then a \emph{pseudorandom generator} (PRG) for $\mathcal{F}$ with error $\eps>0$ is a random variable $\mathbf{X}\in \{-1,1\}^n$ such that for all $f\in \mathcal{F}$,
\begin{equation*}
    \vert \mathbb{E}_{\mathbf{X}}[f(\mathbf{X})]-\mathbb{E}_{\mathbf{U}_n}[f(\mathbf{U}_n)]\vert \leq \eps,
\end{equation*}
where $\mathbf{U}_n$ is the uniform distribution on $\{-1,1\}^n$. If $\mathbf{X}=G(\mathbf{U}_s)$ for some explicit function $G:\{-1,1\}^s\to \{-1,1\}^n$, then $\mathbf{X}$ has \emph{seed length} $s$.
\end{definition}

\begin{definition}\label{def:fprg}
A \emph{fractional pseudorandom generator} (fractional PRG) for $\mathcal{F}$ with error $\eps>0$ is a random variable $\mathbf{X}\in [-1,1]^n$ such that for all $f\in \mathcal{F}$ (identifying $f$ with its multilinear expansion)
\begin{equation*}
    \vert \mathbb{E}_{\mathbf{X}}[f(\mathbf{X})]-f(\mathbf{0})\vert \leq \eps,
\end{equation*}
where the definition of seed length is the same. A fractional PRG is $p$-\emph{noticeable} if for each $i\in [n]$, $\mathbb{E}[\mathbf{X}_i^2]\geq p$.
\end{definition}

We now state the main results of \cite{CHHL} and \cite{CHLT} that show how to construct PRGs from suitably combining noticeable fractional PRGs. This is done by the following \emph{amplification theorem}, which roughly composes fractional random variables into a random walk inside the Boolean hypercube:

\begin{THM}
\label{thm:amplification}
Suppose $\mathcal{F}$ is class of $n$-variate Boolean functions that is closed under restrictions, and that $\mathbf{X}$ is a $p$-noticeable fractional PRG with error $\eps$ and seed length $s$. Then there exists an explicit PRG for $\mathcal{F}$ with seed length $O(s\log(n/\eps)/p)$ and error $O(\eps\log(n/\eps)/p)$.
\end{THM}

Using this result, \cite{CHHL} proved the following theorem that exploits strong $L_1$ control of each Fourier level:

\begin{THM}
\label{thm:chhl}
Let $\mathcal{F}$ be any class of $n$-variate Boolean functions that is closed under restrictions. Suppose that $L_{1,k}(\mathcal{F})\leq b^k$ for some $b\geq 1$ and all $1\leq k\leq n$. Then for any $\eps>0$, there exists an explicit PRG for $\mathcal{F}$ with error $\eps$ and seed length $b^2\cdot\mathrm{polylog}(n/\eps)$.
\end{THM}

This is achieved by constructing a fractional PRG that is a scaled version of a nearly $\log(1/\eps)$-wise independent distribution. As we will be analyzing a similar fractional PRG, we defer the details to next section. To lessen the requisite assumptions on the Fourier spectrum, Chattopadhyay {et al.} \cite{CHLT} derandomize a construction of Raz and Tal \cite{RazTal} to prove the following result that requires only level-two control, albeit at a cost of exponentially worse dependence on the error $\eps$, and quadratically worse dependence on the level-two mass:

\begin{THM}
\label{thm:chlt}
Let $\mathcal{F}$ be any class of $n$-variate Boolean functions that is closed under restrictions. Suppose that $L_{1,2}(\mathcal{F})\leq b^2$ for some $b\geq 1$. Then for any $\eps>0$, there exists an explicit PRG for $\mathcal{F}$ with error $\eps$ and seed length $O((b^2/\eps)^{2+o(1)}\mathrm{polylog}(n))$.
\end{THM}

\section{Low-Degree Polynomial Approximations on Subcubes}
\label{sec:lowdeg}

Throughout this section, we assume that $\mathcal{F}$ is a class of $n$-variate Boolean functions closed under restrictions.  As mentioned above, the main result from which we derive our improvements in constructing pseudorandom generators is essentially a statement about low-degree polynomial approximations on subcubes $[-c,c]^n$ for $c<1$.   We remark that this setting is equivalent to approximating \emph{noisy} versions $T_cf$ on $[-1,1]^n$, where $T_{\rho}$ is the $\rho$-noise operator.  This is because for any $\mathbf{y} \in [-c,c]^n$, we can write $\mathbf{y} = c\mathbf{x}$ for some $\mathbf{x} \in [-1,1]^n$ and
\[
f(\mathbf{y})
=f(c\mathbf{x})
= \sum_{S \subseteq [n]} \hat{f}(S) (c\mathbf{x})^S
= \sum_{S \subseteq [n]} c^{|S|} \hat{f}(S) \mathbf{x}^S
= T_cf(\mathbf{x}) .
\]
In general, given any $k\leq n$, $c\geq 0$, and any $f\in \mathcal{F}$, let $\eps_{c,k}(f)$ be defined by
\begin{equation}
\label{eq:fapprox}
    \eps_{c,k}(f) \triangleq \inf_{g: \mathrm{deg}(g)< k}\max_{\mathbf{x}\in [-c,c]^n}  \vert f(\mathbf{x})-g(\mathbf{x})\vert,
\end{equation}
and extend the definition to function classes by
\begin{equation*}
    \eps_{c,k}(\mathcal{F}) \triangleq \max_{f\in \mathcal{F}} \eps_{c,k}(f).
\end{equation*}
Now, given $\eps>0$, $k\leq n$, and the class $\mathcal{F}$, define $c_k(\mathcal{F},\eps)$ by
\begin{equation*}
    c_k(\eps,\mathcal{F})\triangleq \max\{ c\geq 0: \eps_{c,k}(\mathcal{F})\leq \eps\}.
\end{equation*}
In words, $c_k(\eps,\mathcal{F})$ measures how small a hypercube we must take to ensure that for every function in our class, there exists a degree-($k-1$) approximating polynomial that agrees with $f$ up to a uniform $\eps$ error on the subcube $[-c,c]^n$; by multilinearity, it actually suffices that this holds at the extreme points $\{-c,c\}^n$. Note that \cref{eq:fapprox} can be formulated as a linear program and its optimal solution is the best low-degree  $\ell_{\infty}$-approximation to $f$.

The main technical claim in this section is that we bound $c_k \left( \eps , \mathcal{F} \right)$ in terms of $M_k(\mathcal{F})$.  Specifically, we show that for any class $\mathcal{F}$ that is closed under restrictions, truncating the Fourier expansion of a function $f \in \mathcal{F}$ to its first $(k-1)$ levels serves as a good approximation to $f$ on a sufficiently small hypercube around the origin.

\begin{THM}
\label{thm:main}
        Let $f\in \mathcal{F}$ that is closed under restrictions. Then for all $c \in (0,1)$, we have     \begin{equation*}
        \max_{\mathbf{x} \in [-c , c]^n } \left \vert f_{\geq k} (\mathbf{x}) \right\vert \leq \left( \frac{c}{1 - c} \right)^k M_k(\mathcal{F}).
    \end{equation*}
    In particular, it follows that
    \begin{equation*}
        \eps_{c,k}(\mathcal{F})\leq \left( \frac{c}{1 - c} \right)^k M_k(\mathcal{F}).
    \end{equation*}
\end{THM}
From \Cref{thm:main}, one immediately obtains a lower bound on $c_k(\eps,\mathcal{F})$:
\begin{corollary}
\label{cor:ckbound}
For any class $\mathcal{F}$ that is closed under restrictions, and any $\eps>0$ and $k\leq n$, 
\begin{equation*}
    c_k(\eps,\mathcal{F}) = \Omega\left(\left(\frac{\eps}{M_k(\mathcal{F})}\right)^{1/k}\right)
\end{equation*}
\end{corollary}
\begin{proof}
Observe that by setting $c=\Omega\left(\left(\frac{\eps}{M_k(\mathcal{F})}\right)^{1/k}\right)$ in \Cref{thm:main}, the right side is bounded by $\eps$. Because $f_{\geq k}=f-f_{<k}$ and $f_{<k}$ has degree strictly less than $k$, it follows immediately from the definition of $c_k(\eps,\mathcal{F})$ that $c_k(\eps,\mathcal{F})$ is at least $c$.
\end{proof}
    We now return to the proof of \Cref{thm:main}. To prove this result, we require the following intermediate claims. The first simply shows that we may always bound the contribution of the level-$k$ part of any function in $\mathcal{F}$ by simply rescaling the argument: 
    \begin{lemma} 
    \label{lem:mult}
        Let $f \in \mathrm{conv}(\mathcal{F})$. 
        Then, for all $c\in ( 0,1 )$ and $\mathbf{x} \in [-c , c  ]^n$, we have 
        \begin{equation*}
            \vert f_k(\mathbf{x})\vert \leq c^k M_k(\mathcal{F}). 
        \end{equation*}
    \end{lemma}
    
    \begin{proof}
    Observe that $c^{-1}\mathbf{x}\in [-1,1]^n$ by assumption, and by homogeneity of $f_{k}$ as a polynomial, we have
    \begin{equation*}
        \vert f_k(\mathbf{x})\vert=c^k \vert f_k(c^{-1}\mathbf{x})\vert\leq c^k M_k(\mathrm{conv}(\mathcal{F}))= c^k M_k(\mathcal{F}). \qedhere
    \end{equation*}
    \end{proof}

The next simple yet powerful claim shows that one can ``recenter'' functions in $\mathcal{F}$ and they remain in $\mathrm{conv}(\mathcal{F})$ (and therefore, enjoy the same Fourier bounds). This random restriction technique is a key tool in \cite{CHHL}.

    \begin{lemma} \label{lem:conv}
        Let $f \in \mathrm{conv}(\mathcal{F})$, $\mathbf{a} \in [-1,1]^n$ and $\mathbf{b} \in [0,1]$ such that $ \vert a_i\vert + b_i \leq 1$ for all $i\in [n]$. 
        Define $\tilde{f} $ by $\tilde{f}(\mathbf{x}) = f(\mathbf{a} + \mathbf{b} \circ \mathbf{x})$, where $\circ$ denotes componentwise multiplication. Then, $\tilde{f}  \in \mathrm{conv}(\mathcal{F})$. 
    \end{lemma}
    \begin{proof}
        Given $\mathbf{a}$ and $\mathbf{b}$, define a distribution $D_i$ on $Z_i= \{ -1,1 , x_i\} $ where $x_i$ is treated as formal variable, such that $\mathbb{E}_{y_i \sim D_i }[ y_i] = a_i + b_i x_i  $; note that this is possibly by the assumption that $\vert a_i\vert + b_i \leq 1$. 
        Let $D = \prod_i D_i  $ be the product distribution of the $D_i$. For any $\mathbf{z} \in \prod_i Z_i $, define $f_{\mathbf{z}}(\mathbf{x})$ as the function obtained by setting $x_i=z_i$ for each $i$; in particular, each variable gets set to $\pm 1$ or remains a formal variable.  By our assumption on the closure of $\mathcal{F}$, we clearly have $f_{\mathbf{z}}\in \mathcal{F}$ for any $\mathbf{z}$. By multilinearity and independence of the product distribution, we have $f(\mathbf{a} + \mathbf{b} \circ \mathbf{x}) = \mathbb{E}_{\mathbf{z} \sim D}[ f_{\mathbf{z}}(\mathbf{x})]  $. Thus $\tilde{f} \in \mathrm{conv}(\mathcal{F})$. 
    \end{proof}

As mentioned before, our approach will be to bound the higher-order terms of the Fourier expansion at the fractional points of the fractional PRG via the error term that arises in Taylor's theorem. Denote by $h^{(k)}$ the $k$-th derivative of any $C^k$ function $h:\mathbb{R}\to \mathbb{R}$. We then have the following claim:
    \begin{lemma} 
    \label{lem:derivative}
        Let $f:\mathbb{R}^n \to \mathbb{R}  $ be multilinear and let $\mathbf{x} \in \mathbb{R}^n $. 
        Define $g:\mathbb{R} \to \mathbb{R}$ by $g(t) = f(t \mathbf{x}) $. 
        Then, 
        \begin{equation*}
            g^{(k)}(0) = k! \cdot f_k ( \mathbf{x}). 
        \end{equation*}
    \end{lemma}
    \begin{proof}
        From the definition, it follows that 
        \begin{equation*}
            g( t) = \sum_{S \subseteq [ n ] } t^{\vert S\vert} \hat{f}(S) \mathbf{x}^S.
         \end{equation*}
         Differentiating $g$ with respect to $t$, we get 
         \begin{equation*}
            g^{(k)}(t)=\sum_{S: \vert S\vert  \geq k} \Bigl(\, \prod_{i=0}^{k-1}( \vert S\vert-i) \Bigr) t^{ \vert S\vert -k} \hat{f}(S) \mathbf{x}^{S}. 
         \end{equation*}
         Setting $t = 0$ eliminates all of the monomials with $\vert S\vert > k $, giving us the required bound. 
    \end{proof}

The last intermediate result we require connects the function defined in the previous part with our assumed Fourier bounds:
    \begin{lemma}
    \label{lemma:finallem}
        Let $f \in \mathrm{conv}(\mathcal{F}) $, $c \in (0,1) $ and $\mathbf{x} \in [ -c , c]^n $. 
        Define $g$ as in \cref{lem:derivative}. 
        Then, 
        \begin{equation*}
            \max _{s \in[0,1]} \bigl| g^{(k)}(s) \bigr|
            \leq\left(\frac{c}{1-c}\right)^{k} \cdot k! \cdot M_{k}(\mathcal{F})
        \end{equation*}
    \end{lemma}
    \begin{proof}
        Fix $s \in [ 0,1]$ and let $\lambda=1-c \in [0,1]$.  Define the auxiliary function $\tilde{f}(\mathbf{y})=f(s\mathbf{x}+\lambda\mathbf{y})$. Writing $\mathbf{a}=s\mathbf{x}$ and $\mathbf{b}=(\lambda,\ldots,\lambda)$, we clearly have $s |x_i| + \lambda \leq 1$, so we may apply \cref{lem:conv} to see that $\tilde{f}\in \mathrm{conv}(\mathcal{F})$.
        Now writing $\tilde{g}(t) = \tilde{f}(t\mathbf{x})=f(s\mathbf{x}+\lambda t\mathbf{x})$, we also have $\tilde{g}(t)=g(s+t\lambda)$. By the chain rule, differentiating both sides $k$ times and then setting $t=0$, we have
        \begin{equation*}
            \lambda^k g^{(k)}(s)=\tilde{g}^{(k)}(0).
        \end{equation*}
        On the other hand, by \cref{lem:derivative}, we have $\tilde{g}^{(k)}(0)=k!\cdot \tilde{f}_k(\mathbf{x})$, and as $\tilde{f}\in \mathrm{conv}(\mathcal{F})$ by \cref{lem:conv}, we conclude using \cref{lem:mult} that
        \[
            \bigl|g^{(k)}(s)\bigr|=\left|\frac{\tilde{g}^{(k)}(0)}{\lambda^k}\right|\leq \left(\frac{c}{1-c}\right)^k \cdot k!\cdot M_k(\mathcal{F}) . \qedhere
        \]
    \end{proof}
    With these intermediate claims taken care of, we may now put them together to obtain \Cref{thm:main}.
    \begin{proof}[Proof of \cref{thm:main}]
    The second statement follows immediately from the first by setting $g=f_{<k}$ for any given $f$, and noticing that $f-g=f_{\geq k}$. Therefore, we focus on the first statement.
    
Let $f\in \mathcal{F}, \mathbf{x}\in [-c,c]^n$ and define $g(t)=f(t\mathbf{x})$. Then, by Taylor expanding $g$ about $t=0$ and evaluating $g$ at $t=1$, we have

\begin{align} \label{eq:decomp-g}
    g(1)=\sum_{i<k} \frac{g^{(i)}(0)}{i!} +R_k, 
\end{align}
where $R_k$ is the error term and is given in Lagrange form by
\begin{equation*}
    R_k = \frac{g^{(k)}(s)}{k!}
\end{equation*}
for some $s\in (0,1)$. By \cref{lem:derivative}, we easily see that the first term in the right hand side of \cref{eq:decomp-g} is precisely $f_{<k}(\mathbf{x})$, and as $g(1)=f(\mathbf{x})$, we clearly then must have $R_k = f_{\geq k}(\mathbf{x})$.  Therefore, by \cref{lemma:finallem}, we obtain
\begin{equation*}
    \vert f_{\geq k}(\mathbf{x})\vert = \left\vert\frac{g^{(k)}(s)}{k!}\right\vert \leq \left(\frac{c}{1-c}\right)^k M_k(\mathcal{F}),
\end{equation*}
as desired.
\end{proof}

\subsection{Lower Bounds via Chebyshev Polynomials}

In this subsection, we show that our bounds on the uniform error of any low-degree polynomial approximator are essentially tight for a reasonable range of $c<1$. Recall that \Cref{thm:main} shows that the low-degree Fourier expansion is an excellent approximator to the original function for $c$ small enough; we now show that this bound cannot be significantly improved for a reasonable range of $c$ using \emph{any} approximator. Our main result of this section is the following converse:

\begin{THM}
\label{thm:cheby}
Let $\mathcal{F}$ be any class of $n$-variate multilinear functions that are closed under restrictions.  Then for any $c\leq \min\left(\frac{1}{3}, 3^{-k} \frac{M_k(\mathcal{F})}{M_{k+1}(\mathcal{F})}\right)$, we have
\begin{equation*}
    \eps_{c,k}(\mathcal{F}) \geq \left(\frac{c}{2}\right)^k M_k(\mathcal{F}).
\end{equation*}
\end{THM}

Recall that on the interval $[-1,1]$, the Chebyshev polynomials give the minimum $\ell_{\infty}$ norm among all polynomials with same leading coefficient in magnitude:

\begin{fact}[Theorem 1.5.4 of \cite{rahman2002analytic}]
\label{fact:cheb3}
If a polynomial $f:\mathbb{R}\to \mathbb{R}$ is monic of degree $n$, then $\max_{x\in [-1,1]} \vert f(x)\vert\geq 2^{-n+1}$, with equality if and only if $f=T_n$, the normalized $n$-th Chebyshev polynomial.
\end{fact}

\begin{proof}[Proof of \Cref{thm:cheby}]
Let $(f,\mathbf{x})$ attain the maximum in the definition of $M_k(\mathcal{F})$, namely
\begin{equation*}
    M_k(\mathcal{F}) = \left\vert \sum_{S\subseteq [n]:\vert S\vert=k} \widehat{f}(S)\mathbf{x}^S\right\vert.
\end{equation*}
First, note that the claim is trivial if every function in $\mathcal{F}$ is of degree at most $k$, because then $f_{\geq k}$ is a homogeneous polynomial of degree $k$ and this lower bound is trivial. Under this assumption, $M_{k+1}(\mathcal{F})>0$. Fix $c\in (0,1)$ and let $p:[-1,1]^n\to \mathbf{R}$ be any multilinear polynomial of degree strictly less than $k$. Define the univariate function $g:[-1,1]\to \mathbb{R}$ by
\begin{equation*}
    g(t) = f(tc\mathbf{x})- p(tc\mathbf{x}).
\end{equation*}
By taking the Fourier expansion of $f$, it is easy to see that the coefficient of $t^{\ell}$ for $\ell\geq k$ is precisely
\begin{equation*}
    c^{\ell}\sum_{S\subseteq [n]:\vert S\vert=\ell} \widehat{f}(S)\mathbf{x}^S,
\end{equation*}
so that the coefficient of $t^k$ is equal to $c^k M_k(\mathcal{F})$ in magnitude.
We then have
\begin{align*}
    \sup_{\mathbf{z}\in [-c,c]^n} \vert f(\mathbf{z})-p(\mathbf{z})\vert&\geq \max_{\mathbf{z}\in [-c\mathbf{x},c\mathbf{x}]} \vert f(\mathbf{z})-p(\mathbf{z})\vert\\
    &=\sup_{t\in [-1,1]} \vert g(t)\vert\\
    &\geq \sup_{t\in [-1,1]} \vert g_{\leq k}(t)\vert - \sup_{t\in [-1,1]} \vert g_{\geq k+1}(t)\vert.
\end{align*}
By \Cref{fact:cheb3}, the first term is at least $c^kM_k(\mathcal{F})/2^{k-1}$. On the other hand, the second term can be bounded using \Cref{thm:main} by
\begin{equation*}
    \sup_{t\in [-1,1]} \vert g_{\geq k+1}(t)\vert\leq \left(\frac{c}{1-c}\right)^{k+1}M_{k+1}(\mathcal{F}).
\end{equation*}
Therefore, we obtain
\begin{equation*}
    \sup_{\mathbf{z}\in [-c,c]^n} \vert f(\mathbf{z})-p(\mathbf{z})\vert\geq 2\left(\frac{c}{2}\right)^kM_k(\mathcal{F})-\left(\frac{c}{1-c}\right)^{k+1}M_{k+1}(\mathcal{F}).
\end{equation*}
It is straightforward to verify that for $c\leq \min\left(1/3, 3^{-k}  \frac{M_k(\mathcal{F})}{M_{k+1}(\mathcal{F})}\right)$, the second term is bounded by half of the first. Because $p$ was an arbitrary low-degree multilinear polynomial, the claim follows.
\end{proof}

    \section{From Polynomial Approximations to PRGs}
    \label{sec:patoprg}
\subsection{From Polynomial Approximations to Fractional PRGs}
From \Cref{thm:main}, we now show how the construction of fractional PRGs from level-$k$ bounds reduces to efficient polynomial approximation on ``large'' subcubes.
    \begin{THM}
    \label{thm:mkfprg}
        Let $\mathcal{F}$ be closed under restrictions. Then there exists a fractional PRG for $\mathcal{F}$ with error $\eps$ and seed length $O(k\log n)$ that is $\left(c_k(\eps/2,\mathcal{F})\right)^2$-noticeable. In particular, if $M_k(\mathcal{F})=b^k$, there exists such a fractional PRG that is $\Omega\left(\frac{\eps^{2/k}}{b^2}\right)$-noticeable with seed length $O(k\log n)$.
\end{THM}

\begin{proof}
The second statement follows immediately from the first using \Cref{cor:ckbound}, so we focus on the first statement.

Fix $f\in \mathcal{F}$, $\eps>0$, and let $\mathbf{X}$ be a $(k-1)$-wise independent random variable over $\{-1,1\}^n$ such that $\vert \mathbf{X}_i\vert=c \leq 1/2$ for all $i \in [n]$ for some $c>0$ we specify momentarily. It is well-known that $\mathbf{X}$ can be sampled efficiently with seed length $O(k\log n)$ \cite{vadhan2012pseudorandomness}. By definition of $c:=c_k(\eps/2,\mathcal{F})$, there exists a degree-$(k-1)$ multilinear polynomial $\widetilde{f}$ which $\eps$-approximates $f$ on the subcube $[-c,c]^n$, i.e.
\begin{equation}
\label{eq:tildef}
            \max_{ y \in [-c,c]^n } \bigl| f(y) - \widetilde{f}(y) \bigr| \leq \eps/2.
\end{equation}
Then we have, via the Fourier expansion of $f$,
\begin{align*}
    \bigl| \mathbb{E}_{\mathbf{X}}[f(\mathbf{X})]-f(\mathbf{0})\bigr|
    &\leq \frac{\eps}{2} + \left\vert \mathbb{E}_{\mathbf{X}}[f(\mathbf{X})]-\widetilde{f}(\mathbf{0})\right\vert\\
    &=\frac{\eps}{2} + \left\vert \mathbb{E}_{\mathbf{X}}\bigl[f(\mathbf{X})-\widetilde{f}(\mathbf{\mathbf{X}})\bigr]\right\vert\\
    &\leq \frac{\eps}{2} +  \mathbb{E}_{\mathbf{X}}\left[ \bigl| f(\mathbf{X})-\widetilde{f}(\mathbf{\mathbf{X}})\bigr| \right]\\
    &\leq \eps.
\end{align*}
The first inequality applies \cref{eq:tildef} at the point $\mathbf{x}=\mathbf{0}$, and the second uses the fact that $\mathbf{X}$ is $(k-1)$-wise independent and $\widetilde{f}$ has degree at most $k-1$. The final inequality holds because of (\ref{eq:tildef}) and the fact that $\mathbf{X}\in [-c,c]^n$.  Therefore, $\mathbf{X}$ satisfies the definition of a fractional PRG.  Note that by construction, $\mathbf{X}$ is $c^2$-noticeable since it takes values in $\{-c,c\}^n$.
\end{proof}

Although it does not fit so neatly in this approximation framework, one can essentially recover the improved seed length of \cite{CHHL} (which we recall assumes $L_{1,i}(\mathcal{F})$ bounds for $i=1,\ldots, n$) if one further has $L_{1,i}(\mathcal{F})$ bounds just up to level $k-1$:
\begin{THM}
\label{thm:l1bounds}
        Let $\mathcal{F}$ be closed under restrictions, and suppose that $M_k(\mathcal{F})\leq b^k$ for some $b\geq 1$, $k>2$. If it further holds that $L_{1,i}(\mathcal{F})\leq b^{i}$ for all $1\leq i<k$, then there exists a $\Theta(\eps^{2/k}/b^2)$-noticeable fractional pseudorandom generator for $\mathcal{F}$ with error $\eps$ and seed length $O(\log\log n +\log k+\log(1/\eps))$.
    \end{THM}
    \begin{proof}
    Fix $f\in \mathcal{F}$, and let $\mathbf{X}$ be a random variable such that $\vert \mathbf{X}_i\vert=c$ for all $i \in [n]$ for some $c>0$ we specify momentarily. Then we have, via the Fourier expansion of $f$,
\begin{equation*}
    \bigl| \mathbb{E}_{\mathbf{X}}[f(\mathbf{X})]-f(\mathbf{0}) \bigr|
    = \Biggl| \mathbb{E}_{\mathbf{X}} \Biggl[ \sum_{S\subseteq [n]:1\leq \vert S\vert\leq k-1}\hat{f}(S)\mathbf{X}^S \Biggr] \Biggr|
    + \bigl| \mathbb{E}_{\mathbf{X}}[f_{\geq k}(\mathbf{X})] \bigr| .
\end{equation*}
We first deal with the second term on the right hand side.  By \cref{thm:main} we have
\begin{equation*}
    \bigl| \mathbb{E}_{\mathbf{X}}[f_{\geq k}(\mathbf{X})] \bigr|
    \leq \left(\frac{c}{1-c}\right)^k M_k(\mathcal{F}).
\end{equation*}
By assumption, $M_k(\mathcal{F})\leq b^k$ for some $b\geq 1$; therefore, by taking $c=\Theta(\eps^{1/k}/b)$, this term is at most $\eps/2$. To deal with the first term, we take the same approach as \cite{CHHL}. Under the assumption $L_{1,i}(\mathcal{F})\leq b^i$ for all $i<k$, one may apply their analysis by letting $\mathbf{X}=c\cdot \mathbf{Y}'$, where $\mathbf{Y}'$ is an $(\eps/2)$-almost $(k-1)$-wise independent random variable over $\{-1,1\}^n$.  It is clear that $\mathbf{X}$ is $c^2 = \Theta(\eps^{2/k}/b^2)$-noticeable.  Moreover, exactly as in \cite{CHHL}, we have
\begin{equation*}
    \Biggl| \mathbb{E}_{\mathbf{X}}\Biggl[\sum_{S\subseteq [n]:1\leq \vert S\vert\leq k-1}\hat{f}(S)\mathbf{X}^S\Biggr] \Biggr|
    \leq \sum_{i=1}^{k-1} c^i\sum_{S:\vert S\vert=i} \bigl| \hat{f}(S) \bigr| \bigl| \mathbb{E}[\mathbf{Y}^{'S}] \bigr| 
    \leq (\eps/2) \sum_{i=1}^{k-1} (bc)^i\leq \eps/2,
\end{equation*}
because by our choice of $c$ we have $bc \leq 1/2$. By standard constructions, $\mathbf{Y}'$ can be efficiently sampled with seed length $O(\log\log n +\log k+\log(1/\eps))$ \cite{naor1993small}.  Combining these two errors proves the theorem.
\end{proof}

        \subsection{From Fractional PRGs to PRGs}
\label{sec:fprg2prg}
Using \Cref{thm:mkfprg} and \Cref{thm:l1bounds} in tandem with \Cref{thm:amplification}, it is fairly immediate to obtain PRGs that rely only on a bound on some $k$-th Fourier level.  Similarly, bounds on levels up to $k$ can be leveraged to get an improved seed length.

\begin{THM}[\Cref{thm:levelkprg}, restated]
\label{thm:combined}
Let $\mathcal{F}$ be any class of $n$-variate Boolean functions that is closed under restrictions. Suppose that $M_{k}(\mathcal{F})\leq b^{k}$ for some $b\geq 1$ and  $k> 2$. Then for any $\eps>0$, there exists an explicit PRG for $\mathcal{F}$ with error $\eps$ with seed length
\[
    O\left(\frac{b^{2+\frac{4}{k-2}}\cdot k \log n \cdot \log^{1+\frac{2}{k-2}}(n/\eps)}{\eps^{\frac{2}{k-2}}}\right).
\]
If it further holds that $L_{1,i}(\mathcal{F})\leq b^{i}$ for all $1\leq i<k$, then the seed length can be improved to 
\[
    O\left(\frac{b^{2+\frac{4}{k-2}}\cdot (\log\log n +\log k+\log(b/\eps))\cdot \log^{1+\frac{2}{k-2}}(n/\eps)}{\eps^{\frac{2}{k-2}}}\right).
\]
\end{THM}

\begin{proof}
By  \Cref{thm:amplification}, given an explicit $p$-noticeable fractional PRG for $\mathcal{F}$ with error $\delta$ and seed length $s$, one immediately obtains an explicit PRG for $\mathcal{F}$ with error $O(\delta \log(n/\delta)/p)$ and seed length $O(s\log(n/\delta)/p)$. 

For the first statement, by our assumption and using the fractional PRG guaranteed by   \Cref{thm:mkfprg}, for any $\delta>0$, we immediately obtain an explicit PRG for $\mathcal{F}$ with error $O(b^2 \delta^{1-2/k}\log(n/\delta))$ and seed length $O(b^2 k\log(n)\log(n/\delta)/\delta^{2/k})$. To get the error below $\eps$, we set
\begin{equation*}
    \delta = \Theta\left(\left(\frac{\eps}{b^2 \log(n/\eps)}\right)^{\frac{k}{k-2}}\right)
\end{equation*}
(the astute reader may notice we implicitly use $b\leq n$ here). This yields a PRG with error $\eps$ and seed length
\begin{equation*}
    O\left(\frac{b^{2+\frac{4}{k-2}}\cdot k \log n \cdot \log^{1+\frac{2}{k-2}}(n/\eps)}{\eps^{\frac{2}{k-2}}}\right).
\end{equation*}
The second statement follows in an identical manner from the improved seed length given in the second part of \Cref{thm:l1bounds} in the case that one has control on the $L_1$ Fourier mass on the lower levels.
\end{proof}
\Cref{cor:polylogerr} is now an immediate consequence of \Cref{thm:combined}; for any desired $\eps > b\cdot \log(n)\cdot 2^{-O(k)}$, one can simply apply \Cref{thm:combined} using level $k = \Theta(\log(b\log(n)/\eps))$ to obtain a PRG for $\mathcal{F}$ with error at most $\eps$ with seed length 
\begin{equation*}
    O(b^2\cdot \log(b\log( n)/\eps)\cdot \log(n/\eps)).
\end{equation*}
Note that for error $\eps = 1/\mathrm{poly}(n)$, one needs bounds only up to level $\Theta(\log n)$ (again, using the fact that $b\leq n$). This also partially answers an open question of \cite{CHLT}, which asks how many levels of Fourier bounds suffice to recover polylogarithmic dependence in $1/\eps$.

\begin{remark} \label{rem:lev-2}
Note that this Taylor's theorem approach does not yield anything nontrivial given bounds just on the second level, unlike the fractional PRG in \cite{CHLT}. This is actually a necessary byproduct of combining this approach with the random walk gadget of \cite{CHHL}. Given only level-two bounds, this approach attempts to use $j$-wise independence for $j<k=2$ and smallness to deal with errors on the high degree terms ($k\geq 2$). However, the trivial random variable that is $\pm \mathbf{1}$ with equal probability is trivially $1$-wise independent, as each component is a uniform random bit, albeit trivially correlated. No matter how we scale them, one can show that composing arbitrarily many independent copies of this random variable via the random walk gadget must necessarily polarize to $\pm \mathbf{1}$ at termination, which clearly cannot fool any nontrivial functions.  
\end{remark}

\section{Low-degree Polynomials over $\mathbb{F}_2$}
\label{sec:apps}

Our analysis recovers all the existing applications of \cite{CHHL} (among them, $\mathbf{AC}^0$ circuits, low-sensitivity functions, and read-once branching programs); indeed, all the classes considered there satisfy $L_1$ Fourier bounds on the entire tail. To our knowledge, our new analysis does not immediately improve the seed lengths obtained there, though it shows that (i) \emph{the seed lengths there can potentially be improved using stronger bounds on $M_k$}, and (ii) \emph{the PRGs there would still have fooled those classes had these Fourier bounds been known only up to some level $k$}.

However, the generality afforded to us by this new analysis allows us to obtain a new PRG for low-degree polynomials over $\mathbb{F}_2$, which addresses an open question of \cite{CHHL} by showing that this framework can handle this class. Indeed, let $\mathcal{F}$ be the set of $n$-variate, degree-$d$ polynomials over $\mathbb{F}_2$. As a preliminary step towards deriving Fourier tail bounds that would imply a nontrivial PRG for this class using their framework, \cite{CHHL} proves the following Fourier bounds:

\begin{proposition}[Theorem 6.1 of \cite{CHHL}]
\label{prop:f2bound}
Let $p\colon\mathbb{F}_2^n\to \mathbb{F}_2$ be a degree-$d$ polynomial, and let $f(\mathbf{x})=(-1)^{p(\mathbf{x})}$. Then $L_{1,k}(f)\leq (k\cdot 2^{3d})^k$.
\end{proposition}
Note that this result cannot be applied to their original analysis, for they require a nontrivial bound at all levels, while this bound is trivial for $k=\Omega(\sqrt{n})$ and any $d$. While  \Cref{thm:chlt} can yield a nontrivial PRG by just applying the level-two bound, the dependence on $1/\eps$ is at least quadratic.\footnote{By applying this Fourier bound at level-two, one can use the fractional PRG of \cite{CHLT} to obtain seed length $2^{O(d)}\mathrm{polylog}(n)/\eps^{2+o(1)}$ using the random walks framework. This gives exponentially worse error dependence compared to our approach.} However, using our new, more flexible analysis, one can obtain a nontrivial PRG with polylogarithmic dependence on the error parameter. Our formal result is the following:

\begin{THM}
Let $\mathcal{F}$ be the class of degree-$d$ polynomials over $\mathbb{F}_2$ on $n$ variables. Then there exists an explicit pseudorandom generator for $\mathcal{F}$ with error $\eps$ and seed length
\begin{equation*}
    2^{O(d)}\cdot \log^3(\log(n)/\eps)\cdot \log(n/\eps).
\end{equation*}
\end{THM}
\begin{proof}
Fix $\eps>0$ and let $k=\Theta(\log(\log(n)/\eps))$. By   \Cref{prop:f2bound}, we have that for all $j\leq k$,
\begin{equation*}
    L_{1,j}(\mathcal{F})\leq \Theta\bigl(\log(\log(n)/\eps)\cdot 2^{3d} \bigr)^j.
\end{equation*}
By setting $b=\Theta(\log(\log(n)/\eps)\cdot 2^{3d})$, we may apply   \Cref{thm:combined} for $\mathcal{F}$ and error $\eps$. Note that $\eps^{-\Theta(1/\log(1/\eps))}=O(1)$, so plugging in this value of $b$, we immediately obtain the desired pseudorandom generator.
\end{proof}
For comparison, the best known construction by Viola \cite{viola2009sum}, obtained by summing $d$ independent copies of a sufficiently good small-bias space, attains seed length $d\cdot \log n + O(d\cdot 2^d \log(1/\eps))$, which for constant $\eps$ and $d$ is within a constant factor of the optimal seed length.  The generator implied by our analysis recovers this polylogarithmic dependence in $n/\eps$, although with slightly worse dependence on $\log n$ and polynomially worse dependence in $\log(1/\eps)$. Neither generator can handle superlogarithmic degree. While this result clearly falls short of the state-of-the-art, we emphasize that this generator is conceptually distinct from the existing constructions, and yet belongs to this generic random walk framework.

Our analysis allows us to exploit known Fourier bounds that are too weak for the existing analyses to obtain polylogarithmic error dependence. In particular, to get a nontrivial pseudorandom generator for polynomials of superlogarithmic degree with nontrivial seed length, our work shows that the following weaker conjecture would suffice to break the logarithmic degree barrier and still achieve polylogarithmic (in $n$) seed length for $\eps=1/\mathrm{poly}(n)$:
\begin{conjecture}\label{con:mk}
Let $\mathcal{F}$ be the class of degree-$d$ polynomials over $\mathbb{F}_2$ on $n$ variables. Then 
\begin{equation*}
    M_{k}(\mathcal{F})\leq (\mathrm{poly}(k,\log n)\cdot 2^{o(d)})^k
\end{equation*}
for $k\leq O(\log n)$.
\end{conjecture}
In fact, we observe that to break the logarithmic degree barrier, it actually suffices that this holds just at level $k=3$, though with poor dependence on $\eps$. Note that this is a significantly weaker conjecture than positing that the same bounds hold for $L_{1,k}(\mathcal{F})$. Moreover, as we explain in the next section, $M_k(\mathcal{F})$ can be controlled using correlation bounds, which are much better studied than $L_1$ Fourier bounds.

\section{Bounds on $M_k(\mathcal{F})$ via Correlation with Shifted Majorities}
\label{sec:corrsec}
As we have seen, our new analysis lets one construct PRGs from the weaker quantity $M_k(\mathcal{F})$.  In this section, we extend the argument of Chattopadhyay, Hatami, Hosseini, Lovett, and Zuckerman \cite{xorlemma} to show how bounds on $M_k(\mathcal{F})$ follow from covariance bounds with certain resilient functions (in particular, shifted majorities). In their paper, they deal with the case of $k=2$; we rather straightforwardly generalize this argument, but stress that the approach is the same as in Section 6 of their paper. To that end, for convenience and consistency with their argument, we adopt their conventions and requisite definitions just for this section. We will now consider Boolean functions written as $f:\{0,1\}^n\to \{0,1\}$. Translating to this notation, for any such Boolean function $f$, let $e(f)(\mathbf{x})\triangleq (-1)^{f(\mathbf{x})}$. Then, letting $F=e(f)$, we now have $\hat{F}(S)=\mathbb{E}_{\mathbf{x}}[F(\mathbf{x})e(\sum_{i\in S} x_i)]$.

\begin{definition}
The \emph{covariance between $f$ and $g$}, where $f,g$ are Boolean is 
\begin{equation*}
    \mathrm{cov}(f,g)\triangleq 
    \bigl| \mathbb{E}[e(f(\mathbf{x}))e(g(\mathbf{x}))] - \mathbb{E}[e(f(\mathbf{x}))]\mathbb{E}[e(g(\mathbf{x}))] \bigr|.
\end{equation*}
The covariance between a function $f$ and a class $\mathcal{G}$ is defined as $\mathrm{cov}(f,\mathcal{G})\triangleq \max_{g\in \mathcal{G}} \mathrm{cov}(f,g)$.
\end{definition}

For any $\mathbf{x}\in \{0,1\}^n$, we write $\vert \mathbf{x}\vert$ for its Hamming weight, i.e. $\sum_{i=1}^n x_i$. For any $a\in \{0,1,\ldots,n\}$, \cite{xorlemma} defines $\mathrm{Maj}_a$ by
\begin{equation*}
    \mathrm{Maj}_a(\mathbf{x})\triangleq \begin{cases}
    1 & \text{if $\vert \mathbf{x}\vert> a$}\\
    0 & \text{otherwise},
    \end{cases}
\end{equation*}
as well as the following associated functions for any $\theta\in [n/2]$:
\begin{equation*}
    \mathrm{Thr}_{\theta}(x)\triangleq \begin{cases}
    (-1)^{\mathrm{Maj}_{n/2}(\mathbf{x})} & \text{if $\big\vert\vert \mathbf{x}\vert-n/2\big\vert> \theta$}\\
    0 & \text{otherwise.}
    \end{cases}
\end{equation*}
We now prove the following lemma relating $M_k(\mathcal{F})$ with covariance bounds against the $k$-$\mathsf{XOR}$s of these functions:

\begin{lemma}[Lemma 6.1 of \cite{xorlemma}, adapted]
\label{lem:corrbounds}
Let $\mathcal{F}$ be any family of $(kn)$-variate Boolean functions that is closed under relabeling and negation of input variables.  Suppose that for any $a_1,\ldots,a_k$ such that $\vert a_i-n/2\vert = O(\sqrt{kn \log n})$ for all $i\in [k]$, and all $f\in \mathcal{F}$, we have for some $t\geq 1$
\begin{equation*}
    \mathrm{cov} \bigl(f, \oplus_{i=1}^k \mathrm{Maj}_{a_i} \bigr)
    \leq \left(\sqrt{\frac{t}{n}}\right)^k,
\end{equation*}
where $\oplus$ denotes the $\mathsf{XOR}$ function.
Then, 
\begin{equation*}
    M_k(\mathcal{F})\leq O \bigl(\sqrt{tk\log n} \bigr)^k.
\end{equation*}
\end{lemma}
To prove this lemma, \cite{xorlemma} uses the following sequence of claims. 
\begin{fact}[Claim $6.2$ in \cite{xorlemma}] 
\label{fact:f1}
 For any $f\in \mathcal{F}$, let $F(\mathbf{x}_1,\ldots,\mathbf{x}_k)=e(f(\mathbf{x}_1,\ldots,\mathbf{x}_k))$. Under the hypotheses of \cref{lem:corrbounds}, for any $1\leq a_1,\ldots,a_k\leq O(\sqrt{kn\log n})$, 
\begin{equation*}
    \biggl| \mathbb{E}_{\mathbf{x}_1,\ldots,\mathbf{x}_k} \Bigl[ \bigl(F(\mathbf{x}_1,\ldots,\mathbf{x}_k)-\mathbb{E}[F] \bigr) \prod_{i=1}^k \mathrm{Thr}_{a_i}(\mathbf{x}_i) \Bigr] \biggr| \leq \left(\sqrt{\frac{t}{n}}\right)^k.
\end{equation*}
\end{fact}

\begin{fact}[Claim 6.3 of \cite{xorlemma}]
\label{fact:f2}
For any $\mathbf{x}\in \{0,1\}^n$, $\sum_{i=1}^n e(\mathbf{x}_i)=2\sum_{1\leq a\leq n/2} \mathrm{Thr}_a(\mathbf{x})$.
\end{fact}

\begin{fact}[Claim 6.4 of \cite{xorlemma}, adapted]
\label{fact:f3}

For any Boolean function $f:\{0,1\}^{kn}\to \{0,1\}$, there exists a $k$-equipartition of $[kn]$ into disjoint sets $S_1,\ldots,S_k$ such that
\begin{equation*}
    \biggl| \sum_{S\subseteq [kn]:\vert S\vert=k} \hat{f}(S) \biggr|
    \leq C^k \biggl| \sum_{i_j\in S_j\, \forall j\in [k]} \hat{f}(\{i_1,\ldots,i_k\}) \biggr|
\end{equation*}
for some absolute constant $C>0$.
\end{fact}
As this fact is not quite identical to that in \cite{xorlemma}, we give an argument here:
\begin{proof}
We use the probabilistic method: let $\mathcal{P}$ be the set of $k$-equipartitions of $[kn]$. Let $T\subseteq [kn]$ of size $k$ be arbitrary; without loss of generality, suppose $T=[k]$. Consider a uniformly random $k$-equipartition $P=S_1\sqcup \cdots \sqcup S_k\in \mathcal{P}$. The probability that each $i\in T$ belongs to a distinct $S_j$ is easily seen to be
\begin{equation*}
    \prod_{i=1}^{k-1} \frac{(k-i)\cdot n}{kn-i}\geq \frac{(k-1)!\, n^{k-1}}{(kn)^{k-1}}=\frac{(k-1)!}{k^{k-1}}=e^{-O(k)},
\end{equation*}
where the last equality uses Stirling's approximation.  By symmetry, let $\alpha\in \mathbb{N}$ be the number of $k$-equipartitions that any arbitrary subset $T$ is in. Then we have
\begin{align*}
    \alpha \, \bigg\vert \sum_{S\subseteq [kn]:\vert S\vert=k} \hat{f}(S)\bigg\vert &= \bigg\vert \sum_{P\in \mathcal{P}} \sum_{i_j\in S_j\, \forall j\in [k]} \hat{f}(\{i_1,\ldots,i_k\})\bigg\vert\\
    &\leq \sum_{P\in \mathcal{P}}\bigg\vert \sum_{i_j\in S_j\, \forall j\in [k]} \hat{f}(\{i_1,\ldots,i_k\})\bigg\vert\\
    &\leq \vert \mathcal{P}\vert \max_{P\in \mathcal{P}} \bigg\vert \sum_{i_j\in S_j\, \forall j\in [k]} \hat{f}(\{i_1,\ldots,i_k\})\bigg\vert. \qedhere
\end{align*}
The first line follows from simple counting, while the second is the triangle inequality. Rearranging, we deduce that (writing $T$ as a generic subset of size $k$)
\begin{align*}
    \bigg\vert \sum_{S\subseteq [kn]:\vert S\vert=k} \hat{f}(S)\bigg\vert&\leq \frac{\vert \mathcal{P}\vert}{\alpha} \max_{P\in \mathcal{P}} \bigg\vert \sum_{i_j\in S_j\, \forall j\in [k]} \hat{f}(\{i_1,\ldots,i_k\})\bigg\vert\\
    &= \Pr_{P\sim \mathcal{P}}(T\in P)^{-1} \max_{P\in \mathcal{P}} \bigg\vert \sum_{i_j\in S_j\, \forall j\in [k]} \hat{f}(\{i_1,\ldots,i_k\})\bigg\vert\\
    &\leq e^{O(k)} \max_{P\in \mathcal{P}} \bigg\vert \sum_{i_j\in S_j\, \forall j\in [k]} \hat{f}(\{i_1,\ldots,i_k\})\bigg\vert.
\end{align*}
\end{proof}

The last fact that is needed can be deduced from the Chernoff bound:
\begin{fact}[Claim 6.5 of \cite{xorlemma}, adapted]
\label{fact:f4}
For any $a\geq \Omega(\sqrt{kn\log n})$, $\mathbb{E}[\vert \mathrm{Thr}_a\vert]\leq O(1/n^k)$.
\end{fact}

With these facts, we can now prove \cref{lem:corrbounds} in an entirely analogous fashion to \cite{xorlemma}:

\begin{proof}[Proof of \cref{lem:corrbounds}]
Fix $f\in \mathcal{F}$, and again write $F(\mathbf{x}_1,\ldots,\mathbf{x}_k)=e(f(\mathbf{x}_1,\ldots,\mathbf{x}_k))$. Let $F'= F-\mathbb{E}[F]$. 
Let $U_j = \{ i :  (j-1)n+1\leq i\leq jn \}$.
Then, possibly after relabelling variables, we have by \cref{fact:f3} that
\[
    \biggl| \sum_{S\subseteq [kn]:\vert S\vert=k} \hat{f}(S)\biggr|
    \leq C^k \biggl| \sum_{ i_j \in U_j , \forall j\in [k]} \hat{f}(\{i_1,\ldots,i_k\})\biggr|,
\]
so we may turn to bounding this latter term. We have
\begin{align*}
    \biggl| \sum_{ i_j \in U_j , \forall j\in [k]} \hat{f}(\{i_1,\ldots,i_k\}) \biggr|
    &= \biggl| \sum_{ i_j \in U_j, \forall j\in [k]} \mathbb{E}\Bigl[F'(\mathbf{x}_1,\ldots,\mathbf{x}_k)\prod_{j=1}^k e\bigl((\mathbf{x}_j)_{i_j}\bigr)\Bigr]\biggr|\\
    &=\biggl| \mathbb{E}\Bigl[F'(\mathbf{x}_1,\ldots,\mathbf{x}_k)\prod_{j=1}^k \Bigl(\sum_{ i_j \in U_j }e\bigl((\mathbf{x}_j)_{i_j}\bigr)\Bigr)\Bigr]\biggr|\\
    &\leq 2^k \sum_{1\leq a_i\leq n/2, \forall i\in [k]} \biggl| \mathbb{E}\Bigl[F'(\mathbf{x}_1,\ldots,\mathbf{x}_k)\prod_{i=1}^k \mathrm{Thr}_{a_i}(\mathbf{x}_i)\Bigr]\biggr|\\
    &\leq 2^k\Biggl(\sum_{1\leq a_i\leq O(\sqrt{kn\log n}), \forall i\in [k]} \biggl| \mathbb{E}\Bigl[F'(\mathbf{x}_1,\ldots,\mathbf{x}_k)\prod_{i=1}^k \mathrm{Thr}_{a_i}(\mathbf{x}_i)\Bigr]\biggr| + O(1)\Biggr)\\
    &\leq 2^k\cdot O\bigl(\sqrt{kn\log n} \bigr)^k\cdot \left(\sqrt{\frac{t}{n}}\right)^k\\
    &= O \bigl(\sqrt{tk\log n} \bigr)^k.
\end{align*}
The first inequality follows from \cref{fact:f2}, the second from \cref{fact:f4}, and the last from \cref{fact:f1}. Because we assumed that $\mathcal{F}$ is closed under negations of input variables and $f\in \mathcal{F}$ was arbitrary, we obtain the desired claim from \Cref{lem:unsigned} after absorbing the constant $C$ above into the implicit constant in this bound.
\end{proof}

\section{Discussion and Open Questions}
In this work, we have given a nearly complete interpolation between the previous PRGs obtained in the polarizing random walk framework by exploiting level-$k$ bounds on the class of functions, thus answering an open question from \cite{CHLT}. We do so by exploiting an alternate Fourier analysis via Taylor's theorem and utilizing multilinearity and random restrictions. This new analysis enables us to  construct PRGs from bounds on the potentially much smaller and better-understood Fourier quantity $M_k(\mathcal{F})$, for any $k \ge 3$. By generalizing the connection established in \cite{xorlemma}, this reduces the problem of constructing PRGs in this framework to proving correlation bounds. Further, we show how to get a PRG with an improved seed length if we have bounds on $L_{1,i}(\mathcal{F})$, for all $i \le k$, where $k\ge 3$. A natural open question along these lines is to obtain the improved seed length using bounds on $M_i(\mathcal{F})$ (instead of $L_{1,i}(\mathcal{F})$) for all $i \le k$. Another natural question is to construct a PRG using bounds on just $M_2$ (recall that \cite{CHLT} gives such a construction using bounds on $L_{1,2}(\mathcal{F})$ and our analysis only gives a non-trivial PRG from bounds on $M_k(\mathcal{F})$ when $k \ge 3$).

Finally, exploiting known level-$k$ bounds for $\mathbb{F}_2$ polynomials, our approach shows that the polarizing random walk framework can yield  pseudorandom generators for the class of $\mathbb{F}_2$ polynomials that is competitive with the state of the art.  As mentioned, we hope this paper gives evidence that stronger Fourier control (perhaps via proving the required correlation bounds) can give better PRGs using this framework, and can also handle classes that were previously not known to be possible.  In particular, we emphasize that proving \Cref{con:mk} even for the case of $k=3$ will lead to PRGs for $\mathbb{F}_2$-polynomials with degree $\omega(\log n)$, a longstanding problem in complexity theory. 

    \bibliographystyle{alpha}
    \bibliography{References}
    
\end{document}